%% file: Proxy_crowdsourcing_-_COMSOC_18.tex
\documentclass{comsoc2018}
\usepackage{times}
\usepackage{helvet}
\usepackage{courier}
\usepackage{amsmath}
\usepackage{amssymb}
\usepackage{amsmath}

\usepackage{amsthm}
\usepackage{mathtools}
\usepackage{blindtext}
\usepackage{graphicx}
\graphicspath{ {images/} }
\usepackage{verbatim}
\usepackage{xcolor}
\usepackage{stmaryrd}
\usepackage{tikz}
\input{defs}

\usepackage{url}
\usepackage{gensymb}
\usepackage{textcomp}
\usepackage{verbatim}
\usepackage{hyperref}
\usepackage{wrapfig}
\usepackage{subcaption}
\usepackage{graphicx}
\newcount\Comments  
\newcommand{\kibitz}[2]{\ifnum\Comments=1{\color{#1}{#2}}\fi}
\newcommand{\rmr}[1]{\kibitz{blue}{[RESHEF:#1]}}
\newcommand{\gal}[1]{\kibitz{orange}{[GAL:#1]}}
\newcommand{\obp}[1]{\kibitz{purple}{[OMER :#1]}}

\newcommand{\calXb}{\cal X_{\bot}}

\newcommand*{\defeq}{\stackrel{\text{def}}{=}}
\title{Efficient Crowdsourcing via Proxy Voting}
\author{Gal Cohensius, Omer Ben Porat, Reshef Meir, Ofra Amir \\
	Technion---Israel Institute of Technology\\
}
\begin{document}		
\begin{abstract}  
Crowdsourcing platforms offer a way to label data by aggregating answers of multiple unqualified workers.  We introduce a \textit{simple} and \textit{budget efficient} crowdsourcing method named Proxy Crowdsourcing (PCS).  PCS collects answers from two sets of workers: \textit{leaders} (a.k.a proxies) and \textit{followers}.  Each leader completely answers the survey while each follower answers only a small subset of it.  We then weigh every leader according to the number of followers to which his answer are closest, and aggregate the answers of the leaders using any standard aggregation method (e.g., Plurality for categorical labels or Mean for continuous labels).  We compare empirically the performance of PCS to unweighted aggregation, keeping the total number of questions (the budget) fixed. We show that PCS improves the accuracy of aggregated answers across several datasets, both with categorical and continuous labels.  Overall, our suggested method improves accuracy while being simple and easy to implement.
\end{abstract}	
\section{Introduction} \label{sec:intro}
Crowdsourcing is the process of completing tasks by soliciting contributions from a large group of people.  In recent years, crowdsourcing has become increasingly prevalent, and is being used for a variety of tasks including data collection and labeling, prediction and open innovation.  While crowdsourcing enables leveraging a large pool of workers, the quality of contributions strongly depends on workers' capabilities and motivation which varies widely across workers. In particular, since crowdsourcing workers are usually not experts and since some workers might not exert effort, answers obtained may be erroneous \cite{kazai2011crowdsourcing,vuurens2011much,wais2010towards}. Thus, increasing the accuracy and efficiency of crowdsourced work is an important area of research~\cite{quinn2011human,allahbakhsh2013quality}.

The overall quality of the crowdsourced task is determined by two factors: the quality of inputs (answers) provided by the individual crowd workers, and the method used to aggregate those answers. Prior work has examined ways to improve the quality of inputs (e.g., by providing performance-based monetary incentives~\cite{yin2015bonus}), as well as at methods for aggregating multi-source noisy estimations toward ``truth discovery''~\cite{li2016survey,hung2013evaluation}.

Aggregation of opinions is closely related to \emph{social choice}, where the opinions or preferences of multiple voters are aggregated using some voting method. Mao et al.~\cite{mao2012social,mao2013better} compared the performance of various voting rules for aggregating rankings, and showed that some of them yield an outcome that is significantly closer to the ground truth than others. 
More generally, Mao et al. argue that while sophisticated statistical methods have been developed for specific human computation tasks, these tools often require ``\emph{significant amount of data, work, knowledge, and computational
resources,}'' (\cite{mao2012social}, p.~141) which highlights the advantage of using simple voting rules.

In this paper we suggest to use a simple heuristic that is based on \emph{proxy voting} to improve the accuracy of aggregated answers in crowdsourcing platforms. Proxy voting allows voters not to vote directly, and instead select a \emph{proxy} (presumably another voter with similar opinions) to vote on their behalf. Several papers analyzed the benefits of proxy voting to society from theoretical and empirical aspects (see related work). In a recent theoretical paper, Cohensius et al.~\cite{cohensius2017proxy} showed that with proxy voting, the preferences of a large population from a given distribution can be represented well with only a small subset of ``active voters''. Therefore, we hypothesize that similar ideas can be applied to the problem of aggregating answers of crowd workers.

In this paper, we suggest a Proxy Crowdsourcing (PCS) policy, where answers are collected from two groups: \textit{leaders} answer the complete set of questions, while \textit{followers} are asked only a random subset of the questions, which are only used to find the ``closest'' leader.
Then, we aggregate the leaders' answers, where each leader is weighted by the number of followers to which she is closest.


The PCS method is not supposed to replace existing aggregation methods, but rather to boost their performance by properly weighing the workers or experts. For example, with PCS, instead of asking (say) 10 workers to answer 100 questions each and take the majority answers, we may ask only 5 workers (the leaders) to answer everything, and use the remaining budget to collect answers from 25 additional workers (the followers) who will answer 20 questions each. The aggregated answer is then calculated by taking the \emph{Weighted Majority} of the 5 leaders.

As such, proxy crowdsourcing has potential to improve the accuracy of crowdsourcing while the total effort by workers remains fixed\footnote{Although we do not tackle this point explicitly, reducing the amount of feedback required from each voter has received some attention in the social choice literature (see, e.g. \cite{lu2011robust,drummond2013elicitation}), as means to lower the cognitive burden on human participants.}.
This paper examines the use of proxy crowdsourcing and identifies conditions under which PCS leads to higher quality outcomes (i.e., outcomes that are closer to the ground truth). 


\subsection{Related work}
\paragraph{Improving the efficiency and quality of crowdsourcing}
The crowdsourcing literature has explored several avenues for improving the quality of workers' inputs. Some works focused on designing monetary incentives for workers to elicit higher quality inputs. For example, different schemes of providing workers bonus payments have been proposed and studied~\cite{yin2014monetary,yin2015bonus,yin2013effects,shah2015double}, and peer prediction rules have been used in payment schemes to incentivize truthful reporting by workers~\cite{kamar2012incentives,gao2014trick,radanovic2016incentives}. Other works have suggested non-monetary incentives, such as engaging workers by increasing their curiosity~\cite{law2016curiosity} or eliciting commitments from workers to provide high quality inputs~\cite{elmalech2017but}. Our approach does not attempt to directly improve the quality of the input provided by workers, but rather aggregates inputs in a way that aims to assign higher weights to higher quality workers. 

Another practice commonly used to increase the quality of crowdsourcing is to monitor the quality of workers' responses by implementing ``gold questions'' with known answers in order to filter workers that fail in those questions~\cite{zhang2012big,gormley2010non}. To improve the effectiveness of this approach, Bragg et al.~\cite{bragg2016optimal} proposed methods for determining the optimal amount of gold questions to include. 
 Some works have further explored the use of planning approaches to dynamically determine when to collect additional inputs from workers based on confidence in the current solution~\cite{dai2013pomdp,kamar2012combining}. Proxy crowdsourcing aims to reduce the number of inputs required without requiring the use of sophisticated planning methods. 


Research in crowdsourcing has also considered different methods of aggregating the responses obtained from individual workers to achieve higher quality results~\cite{hung2013evaluation,kamar2012combining,simpson2013dynamic,venanzi2014community,ma2015faitcrowd}. Most closely related to our approach, Venanzi et al.~\cite{venanzi2014community} identify ``communities'' of workers that are similar to each other when workers only provide answers to part of the question set, and use a Bayesian computation to aggregate the inputs from the different communities. In contrast to this approach, proxy crowdsourcing aims to determine a priori how to utilize a budget (i.e., how many answers to elicit from workers), and does not require complex estimation computations in the aggregation of workers' inputs. 

Importantly, we note that the proxy crowdsourcing approach is complimentary to most of the other approaches for improving the quality and efficiency of crowdsourcing processes, and can be used in conjunction with these approaches.  




For an overview of proxy voting in the context of \emph{preference aggregation} (i.e., when there is no ground truth), see Cohensius et al.~\cite{cohensius2017proxy}.

\subsection{Contribution}
We evaluate PCS versus standard aggregation methods  on a wide range of datasets in several domains: some synthetic, some collected for this paper, and some that were collected for different purposes by other researchers. 

Our evaluation shows that for almost all datasets, it is beneficial to aggregate the crowdsourcing answers using PCS, in the sense that it decreases the expected distance between the aggregated answers and the ground truth. In some datasets the improvement reaches up to 32\%.  

We analyze the reasons for this improvement by looking at the initial distribution of workers' competence, and the distribution of leaders' weights after the aggregation.  We establish a preliminary theoretical result that explains why more accurate answer vectors get higher weights.

\def\calX{\mathcal{X}}
\def\calXb{\calX_\bot}
\newcommand{\tilv}[1]{\tilde{\boldsymbol{#1}}}
\newcommand{\hatv}[1]{\hat{\boldsymbol{#1}}}
\def\NN{\textrm{NN}}

\section{Model}
We denote by $\ind X$ the indicator variable of event $X$. $[k]$ is a shorthand for $\{1,\ldots,k\}$.

Let $X^k$ be a k-dimensional Euclidean space, and $d$ be a distance metric on $X^k$.  Let $\cal X_{\bot} \defeq \cal X \cup \bot$, where $\bot$ is a missing entry. Denote by $\cal X_{\bot}^k$ the space of \textit{partial vectors} of $\cal X^k$.
We shall make use of the tilde notation, i.e. $\vec {\tilde x}$, to indicate that a vector is partial. Denote by $\bot(\tilv x), \top(\vec{\tilde x} ) \subseteq [k]$, respectively, the missing and valid entries of $\tilde {\vec x} \in \cal X_{\bot}^k$. 

We denote by $x^{(j)}$ the $j$-th entry of a vector $\vec x$. Given a (partial) vector $\tilv x$ and subset of indices $A\subseteq \top(\vec{\tilde x} )$, we define $\tilv x(A)=(\tilde x^{(i)})_{i\in A}\in \calX^{|A|}$. 
We extend the pseudo-distance function $d$ to partial vectors by only considering joint entries:
\[
d(\tilv x_1,\tilv x_2 ) \defeq d(\tilv x_1(A),\tilv x_2(A)),
\]
where $A= \top(\tilv x_1)\cap \top(\tilv x_2)$.

Given a multi-set of (complete) vectors $S_M=\{\vec x_1,\dots ,\vec x_m\} $ and a partial vector $\tilv x$, we denote by $\NN(S_M,\tilv x)$ the multi-set of vectors in $S_M$ that are the closest to $\tilv x$ according to $d$. Formally,
\[
\NN(S_M,\tilv x) = \argmin_{j \in M} d\left(\vec x_j,\tilv x\right).
\]


\paragraph{Aggregation with leaders}
An \emph{aggregation rule} is a function $\vec g:\mathbb R^m \times \calX^{k \cdot m}\rightarrow \calX^k$. That is, a function that maps a set of $m$ weighted complete vectors to a single vector $\hatv x=\vec g(\vec w_M, S_M)$.\footnote{More precisely, this is a class of functions, one for every $m$. We only use simple and neutral aggregation rules where the effect of weights is like duplicating voters. That is,  $\vec g((w_1,\ldots,w_i+1,\ldots,w_m),S_M) = \vec g((w_1,\ldots,w_i,1,\ldots,w_m),(\vec x_1,\ldots,\vec x_i,\vec x_i,\ldots,\vec x_m))$.}

An \emph{instance} in the space $(\calX^k,d)$ is a tuple $(S_M,\tilde S_N)$, where: 
\begin{itemize}
\item $S_M = \{\vec x_1,\ldots, \vec x_m\}$ is a set of complete answers (``leaders'');
\item $\tilde S_N = \{\tilv x_1,\ldots,\tilv x_n\}$ is a set of partial answers (``followers'');
\end{itemize}

Given an instance and an aggregation rule $\vec g$, the aggregated answer vector is $\vec g(S_M,\tilde S_N, d) = \vec g(\vec w_M,S_M)$ where $w_j= 1+\sum_{i\in N}\frac{\ind{j\in \NN(S_M,\tilv x_i)}}{|\NN(S_M,\tilv x_i)|}$. That is, the number of followers to which $\vec x_j$ is closest, where each follower divides her voting weight equally among all of her nearest neighbors. 	  

\begin{example}  Consider a space $(\calX^k,d)$ where $\cal X = \{0,1\},k=4$, and $d$ is the Hamming distance. Let $\vec g$ be the Weighted Majority function. Consider
\[
S_M = \{(0,0,0,0),(0,1,0,1),(0,1,1,0),(1,1,1,1) \}, \tilde S_N=\{(\bot,1,\bot, 1), (\bot, 0 ,1,\bot) \}.
\]
Then $\NN(S_M,\tilv x_1) = \{2,4\}$ and $\NN(S_M,\tilv x_2)=\{1,3,4\}$. Thus, the weights are $w_1 = w_3 = 1\frac13, w_2=1\frac12, w_4=1\frac56$.  The aggregated answer is $\hatv x=(0,1,1,1)$. 
\end{example}

\paragraph{Evaluation metrics}Given a \emph{ground truth} vector $\vec x^*$ and an aggregated answer vector $\hatv x$, We call $d(\hatv x, \vec x^*)$ the \emph{aggregated error}.
For every individual worker $j$, we call $d(\vec x_j,\vec x^*)$ the \emph{individual error} of worker~$j$.


\paragraph{Crowdsourcing policies}
Consider a population $\calW$ of workers, which is a distribution over  complete answer vectors $\calX^k$. 

A \emph{crowdsourcing policy with leaders} (or just \emph{policy}) is a tuple $(m,n,\alpha,\vec g,d)$, where $m$ is the number of leaders; $n$ is the number of followers; $\alpha$ is the fraction of answers that each follower should answer; $\vec g$ is an aggregation rule; and $d$ is a distance metric. 

A policy is used first to sample an instance from the population, and then to aggregate answers of this instance as follows. 
The policy samples $m$ full vectors $S_M$ from $\calW$ and $n$ partial vectors $\tilde S_N$ with $\floor{\alpha k}$ valid entries (selected uniformly at random). Then, an aggregated answer $\hatv x= \vec g(S_M, \tilde S_N,d)$ is computed. Note that $\hatv x$ is a random variable. 

E.g., if we have unlimited budget, we can recruit and aggregate a large number of workers to reduce the error. However, we assume that the budget is limited, and that the required budget is linear in the total number of questions we ask. 

To evaluate the benefit of using leaders, we define $PCS_{\alpha,\beta}(B)$ the policy that spends a total budget of $B$, where a fraction $\beta$ from this budget is spent on followers who provide answers on $\alpha k$ of the questions.

Except when specified otherwise, the policies we will use to compare crowdsourcing with and without leaders are $PCS(B) = PCS_{0.2,0.33}(B)$ and $CS(B)=PCS_{0,0}(B)$, respectively. 

For example, if $k=20$ and $B$ is sufficient for $1200$ answers, then $CS(B)$ results in the policy $(m=60,n=0,\alpha=0,\vec g,d)$ (60 workers that answer all questions); whereas
$PCS(B)$ results in the policy $(m=40,n=100,\alpha=0.2,\vec g, d)$. That is, $40$ leaders will answer all $20$ questions, and 
the remaining half of the budget will be spent on $100$ followers who will answer $\floor{\alpha k}=4$  random questions each.


\section{Empirical Methodology}
We evaluated the use of proxy voting on datasets from three sources: 1. datasets we collected using Amazon Mechanical Turk; 2. existing datasets from \cite{shah2015double}; 3. datasets we generated from simple distributions for specified parameters.  Workers were randomly divided into leaders and followers.

These datasets included a variety of tasks differing in the types of questions asked (e.g., multiple choice questions, continuous assessments\obp{multiple choice / categorical}). We next describe these in more detail. 

In the datasets we collected,  participants were given short instructions,\footnote{see survey in \href{https://goo.gl/W47X5R}{https://goo.gl/W47X5R}.}  then they had to answer $k=25$ questions. We recruited participants through Amazon Mechanical Turk. We restricted participation to workers that had at least 50 assignments approved. We plant in each survey a simple question that can be easily answered by a anyone who understand the instructions of the experiment (known as Gold Standards questions).  Participants who answered correctly the gold standard question received a payment of $\$0.3$.  Participants did not receive bonuses for accuracy.  The study protocol was approved by the Institutional Review Board at the Technion.

For each new or existing dataset, we mention below the total number of collected answer vectors. However we emphasize that unless specified otherwise, aggregation was always performed with the $CS(B)$ and $PCS(B)$ policies, where $B$ is sufficient budget for exactly 12 complete answer vectors. For robustness, we tested each policy by sampling 5000 instances with replacements, using the entire dataset as our population $\calW$. For each instance, we computed the distance from the aggregated result $\hatv x$ to the ground truth $\vec x^*$ of the dataset, and averaged the aggregated error over all instances. We refer to the average expected error of a policy as the \emph{loss}.

\subsection{Binary questions}
Each question has two possible answers, $\calX =  \{a,b\}$. Hamming distance was used to calculate distances between vectors. Namely, the distance between two vectors is the number of questions on which they disagree. As an aggregation rule $\vec g$ we use {Weighted Plurality}, herein denoted by $\wpl$.
Formally, 
\[
\wpl(\vec w_M,S_m)^{(j)} \defeq 
\begin{cases}
a & \text{if }\sum_{i\in S_M:x_i^{(j)} = a} w_i > \frac{m}{2} \\
b & \text{if }\sum_{i\in S_M:x_i^{(j)} = b} w_i > \frac{m}{2} 
\end{cases},
\]
and ties are broken uniformly at random. Note that in the binary case, Weighted Plurality and Weighted Majority coincide. The following datasets were examined: \\

\begin{figure*}[t!]
    \centering
    \begin{subfigure}{0.3\textwidth}
        \centering
       	\includegraphics[scale=0.34]{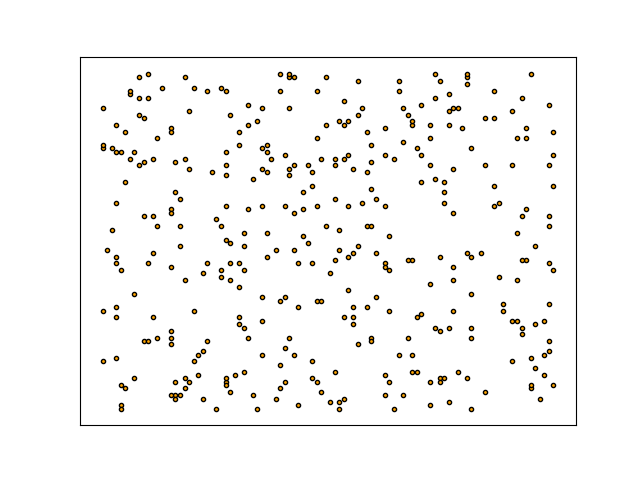}
    \includegraphics[scale=0.34]{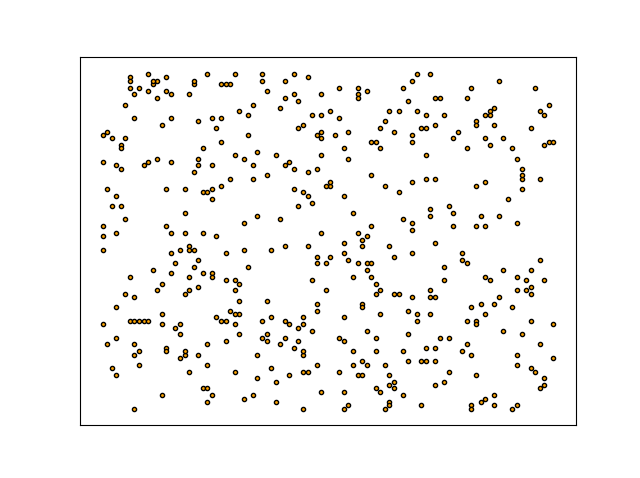}
    \end{subfigure}%
    ~ 
    \begin{subfigure}{0.7\textwidth}
        \centering

        \includegraphics[scale=0.28]{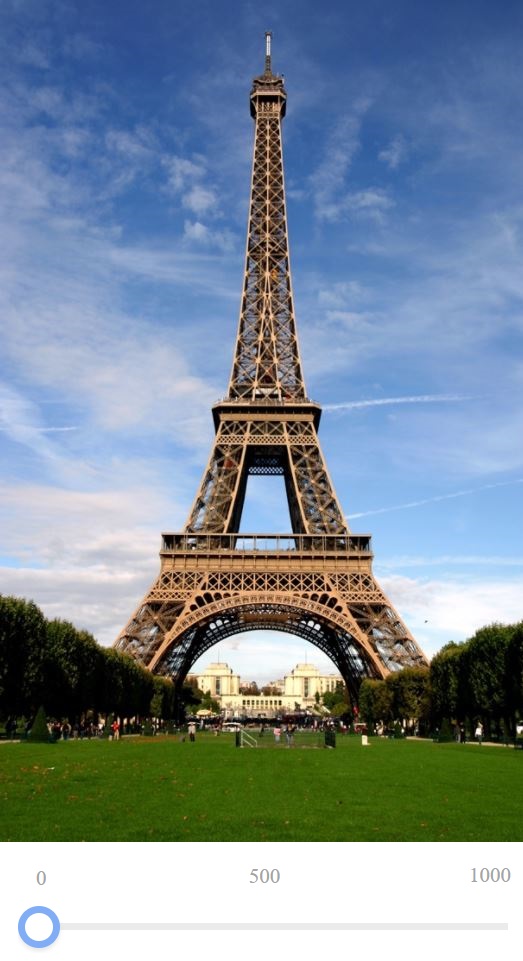}

    \end{subfigure}
    \caption{Representative questions of the Comparing dots (left) and Buildings' height (right) tasks. In the Comparing dots task, each worker was presented with a pair of dot figures, and had to determine which figure contained more dots.  In the Buildings' height task, each worker was given a picture of a building, and had to estimate the building height in meters by selecting a number between 0 and 1000 along a slider.\label{fig:tasks}}
\end{figure*}

\noindent\textbf{Comparing Dots} 41 workers were shown 25 pairs of pictures and had to estimate which picture out of the two has more dots\footnote{A similar task was performed in \cite{mao2013better}, where workers ordered a set of four pictures by increasing number of dots. While we used their dataset in a preliminary version of this paper, we found it unsuitable for the current paper since we wanted datasets where every participant makes more decisions.}

\noindent\textbf{Golden Gate} 35 workers were asked 21 questions of identifying if the picture contains the Golden Gate Bridge.  This dataset was collected by \cite{shah2015double}.  

\subsection{Categorical questions}
  Each question has a discrete set of answers, $\calX =  \{a,b,\dots, l\}$, and the distance function is the hamming distance. As in the Binary domain, we aggregate answers using $\wpl$ defined above when naturally extended to multiple possible answers. We used the following datasets, all collected by \cite{shah2015double}:\\

\noindent\textbf{Dogs} 31 workers answered 85 questions, identifying the breed of the dog in the picture. We omitted 4 workers due to missing data in their report.

\noindent\textbf{Heads of Countries} 32 workers answered 20 questions, identifying heads of countries. We omitted 3 workers due to missing data in their report.

\noindent\textbf{Flags} 35 workers answered 126 questions, identifying countries to which flags belong. 


\subsection{Continuous questions}
$\calX$ is an interval $[0,H]$. The distance metric is the $L^1$-norm, thus $d(\vec x,\vec x') = \sum_{j\leq k}|x^{(j)}-x'^{(j)}|$. We used the \textbf{Weighted Mean}, denoted by $\wmn$, as an aggregation function\footnote{We also considered the Weighted Median, which produced roughly the same results and hence omitted.}.  Formally, 
\[
\wmn(\vec w_M,S_m)_j = \sum_{i\in M}w_i x_i^{(j)}.
\]
We examined the following datasets:\\

\noindent\textbf{Buildings' Height} 208 workers were shown 25 buildings pictures, they were asked to estimate their height in meters. In Figure \ref{fig:tasks} we illustrate one such question.

\noindent\textbf{Counting Dots} 201 workers were shown 25 pictures of dots, they had to estimate the number of dots in each picture.  Counting dots in images has been suggested as a benchmark task for human computation in \cite{mao2013better, horton2010dot}.



\def\wvs{\vphantom{$2^{2^2}$}}
\begin{table*}[t]
	\begin{center}
		\begin{tabular}{|c|l||c|c|c|c|c|}
		\hline
        domain $\cal{X}$	&Task: 				& Workers & k (Questions)& Categories & $E[error]$ &$\sigma(error)$ \\ 
        \hline\hline
        Binary		&Comparing Dots   \wvs  				& 41 & 25 &2& 10.5& 2.86 \\\cline{2-7}
        			&Golden Gate \cite{shah2015double} \wvs & 35 & 21 &2& 4.7 & 3.4 \\
        \hline\hline
           			&Heads of Countries \cite{shah2015double} \wvs 	& 32 & 20  &4& 1.8  & 2.6\\ \cline{2-7}
        Categorical	&Flags \cite{shah2015double} \wvs 				& 35 & 126 &4& 47.8 & 29.6\\ \cline{2-7}
        			&Dogs \cite{shah2015double} \wvs  				& 31 & 85  &10& 11.9 & 11.4
        \\\hline\hline 
        Continuous 	&Buildings' Height 	& 208 & 25 &1001& 5271 	& 3074\\ \cline{2-7}
        			&Counting dots  \wvs& 201 & 25 &1001& 6159 	& 1656
        \\\hline

        \hline
		\end{tabular}
	\end{center}
	\caption{Summary of datasets. For each dataset, we include the population size $|W|$, number of questions $k$, and the mean and variance of the individual error. \label{Table: datasets}}
\label{sec:datasets summary}
\end{table*}

\section{Analysis}
Our main observation is that PCS improves accuracy of the aggregate answers by up to $30\% $ as summarized in Table \ref{Table:results summary}.\\

\def\wvs{\vphantom{$2^{2^2}$}}
\begin{table*}[ht] 
	\begin{center}
		\begin{tabular}{|c|l||c|c|c|}
		\hline
        domain $\cal{X}$     &Task: 		& $CS(12k)$ & $PCS_{0.2,0.33}(12k)$ & improvement [\%]
        \\\hline\hline
        Binary			&Comparing dots   \wvs  							& 8.40 	& 8.56 	& -1.90	 
        				\\\cline{2-5}
        	          	&Golden gate \cite{shah2015double} \wvs  			& 2.16	& 1.88 	& 12.96  	
        \\\hline\hline
        				&Heads of countries\cite{shah2015double} \wvs 		& 0.10 	& 0.09 	& 10.00   
        \\\cline{2-5}
        Categorical		&Flags \cite{shah2015double} \wvs 					& 20.97 & 14.42	& 31.23  
        \\\cline{2-5}
        				&Dogs \cite{shah2015double} \wvs  					& 2.98 	& 2.69	& 9.73 
        \\\hline\hline
        Continuous		&Buildings' height \wvs								& 3315 	& 3075 	& 7.24 	
        				\\\cline{2-5}
          			    &Counting Dots  \wvs 								& 5268.5& 5188 	& 1.53 	
        \\\hline
		\end{tabular}
	\end{center}
	\caption{Results summary.  Expected loss comparison across all our datasets between CS and PCS that spends $\beta=1/3$ of the budget on followers and followers answers on $\alpha=0.2$ of the questions in the survey.  The budget $B$ is sufficient for $12k$ answers in each domain. That is, 12 workers that answer on all the questions or 8 leaders and 40 followers.\rmr{}}
\label{Table:results summary}
\end{table*}

\subsection{Why PCS improves accuracy}
\paragraph{Possible (incorrect) explanation: sampling more workers reduces dependencies}  PCS sample more vectors, thus the dependency between answers decreases.  We conjectured that this dependency decrease might reduce the loss.  We compared CS to a policy were same budget is spent entirely on partial vectors.  This method had the same loss as CS, therefore we rejected this explanation.
\paragraph{Better leaders get higher weight}
The improvement of PCS over CS is explained by the weights of leaders: we argue that better leaders receive higher expected weight. In Figure \ref{fig: proxies_weights } (left), we used the policy $PCS_{0.2,0.375}(16k)$ which means we sample 10 leaders and 30 followers.  We sampled workers from the Flags dataset.  At each instance we sorted the leaders by their individual error and calculated their weights.  The best leader got about twice the weight of worst leader.  The observation that better leaders get higher weight is true across all of our datasets, but its magnitude differs.  Another common observation is that the weakest leader gets higher weight than higher ranked leaders, the reason is that in some populations there is a group of weak workers (worse than random clickers), those weak workers delegate their voting weight with high probability to the weakest leader. 
Figure \ref{fig: proxies_weights } (right) shows that weighting of the leaders improve the outcome substantially.
\begin{figure}[ht]
	\centering
  	\includegraphics[trim = 0mm 0mm 18mm 0mm, scale=0.47]{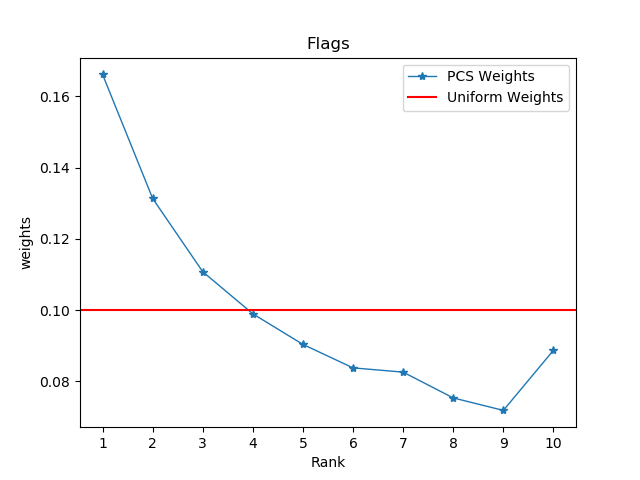}
    \includegraphics[trim = 0mm 0mm 18mm 0mm, scale=0.47]{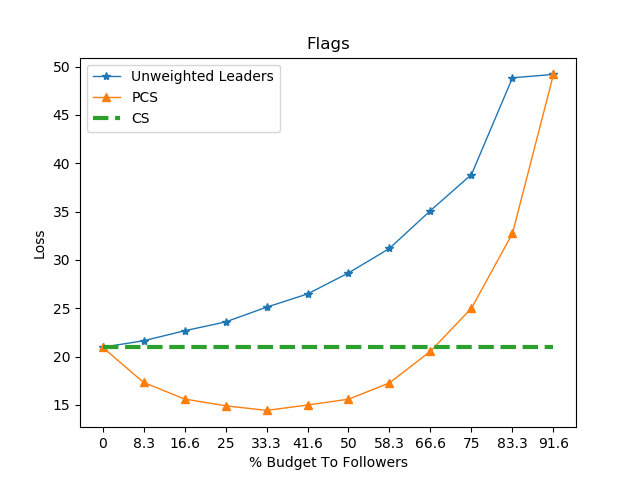}
    \caption{On the left, leaders voting weight sorted by their individual error.  Averaged across 5000 instances, over the Flags dataset.  Rank 1 means the leader with the lowest individual error in the sample, rank 10 means the leader with the highest individual error in the sample.  On the right, the orange triangles shows the loss according to $PCS_{0.2,\beta}(12k)$, for reference, in blue dots, the loss for the same leaders when they are unweighted (only leaders, without followers). \label{fig: proxies_weights }}	
\end{figure}

\subsection{Factors that effect PCS}
\paragraph{Expected loss as a function of the policy}
The expected loss depends on the policy used, that is, the way in which the budget is divided between obtaining leaders and followers. When increasing $\beta$ (the fraction of the budget spent to on followers), the expected loss decreases until reaching an optimum. Further increasing $ \beta$ beyond this point results in higher loss compare to CS. The optimal $\beta $ value depends on the dataset and ranges between $1/3$ and $1/2$.  Higher $\beta$ typically reduces accuracy because a sufficient number of leaders is needed for some high-quality leaders to be present, when some high-quality leader exists, adding followers is more cost effective.

To check this claim we evaluated the expected loss of $PCS_{0.2,\beta}(B)$ when using different $\beta$ values and different budgets over all datasets.  Figure~\ref{fig:PCS_high_improvement} shows the expected loss generated by $CS(12k)$ and $PCS_{0.2,\beta}(12k)$ across different values of $\beta$ on the Flags dataset.  Results were averaged over 5000 instances (similarly to the results in the Table \ref{Table:results summary}).  The results are typical to most of the datasets and most budget values.  Indeed, we see in Figure \ref{fig:PCS_high_improvement} that spending budget on followers is beneficial up to a certain amount (for this dataset and budget, the optimum is at $\beta =1/3$).

\begin{figure}[ht!]
	\centering
 	\includegraphics[scale=0.7]{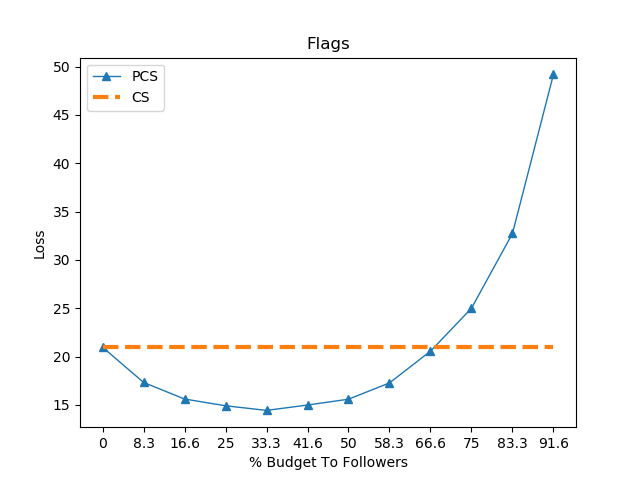}  
    \caption{Expected loss as a function of $\beta$, the fraction of the budget spent to buy followers.  The dashed line is the loss when using the CS policy (all the budget is used to pay for leaders).  The continuous line uses $PCS_{0.2,\beta}(12k)$ (which means $\beta$ of the budget is used to buy followers, and each follower is answering $0.2$ of the questions.  \gal{synthetic dataset that shows the same results, same evaluation to $\alpha$ }
    \label{fig:PCS_high_improvement}}	
\end{figure}

\paragraph{Expected loss as a function of the datasets}
\begin{enumerate}
\item \textbf{Domain}: \textbf{PCS is more effective in Categorical domain rather than Binary domain}.  Table \ref{Table:results summary} provides anecdotal evidence for this claim.  Intuitively, for PCS to work well, accurate workers must be close to other workers while inaccurate workers should be far from other workers.  Consider two workers and one question and assume both workers provide a wrong answer for the question. In the binary domain their distance from each other is $0$, while in the categorical domain they can be wrong in different ways so their distance can be $1$.   \gal{A:compare categorical datasets to the same dataset with unite categorize so 'a' 'b' is 0 and 'c' 'd' is 1,  now see that categorical better reduce the loss.  B: formalize a theorem that prove this}

\item \textbf{Number of questions:}
\textbf{PCS is more effective in datasets with high number of questions}
As the number of questions grows, the probability for incompetent workers to be wrong on the same question (with the same wrong answer) decreases, therefore the probability for a follower to follow incompetent leader decreases.  See section \ref{sec:Explaining leaders Weights} for a preliminary theorem.  In Table \ref{Table:results summary} we see that the dataset with the highest number of question has the highest improvement rate of PCS. \gal{compare dataset after cutting out questions, synthetic: compare across different values of $k$} 

\item \textbf{Individual error Variance}: \textbf{PCS is more effective when the variance of workers' Individual error is high.}  When variance is high, some followers with high individual error can follow leaders with low individual error.  When variance is low, followers will follow similar quality of leaders which will not improve the results.  In order to verify this claim, we produced a synthetic population in the following manner: each worker has a competence parameter $p_i\in (0, 1)$. A worker $i$ with competence $p_i$ answers each question correctly w.p. $p_i$.  Errors are independent across questions and workers.
The competence of each worker was sampled from a uniform distribution over an interval $[v_1, v_2]$,
where $v_1$ and $v_2$ are parameters.

\begin{figure}[ht!]
	\centering
	\includegraphics[  width=12cm,  height=8cm, keepaspectratio ]{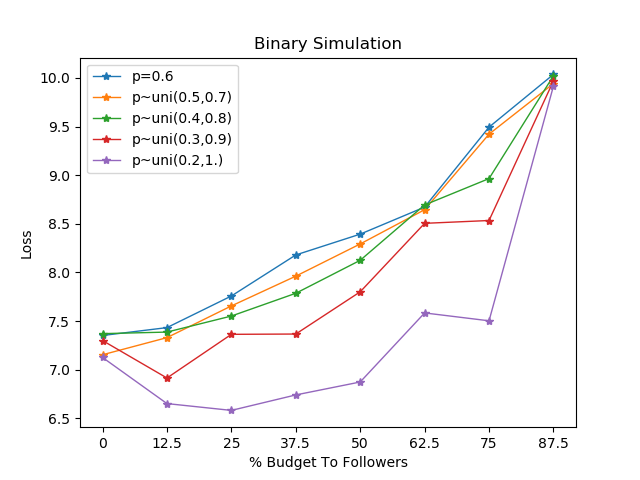} 
	\caption{ Loss as a function of $\beta$.  Synthetic populations with similar mean but different variance of their judgment competence $p$.
	\label{fig:simulation_binary}}	
\end{figure}

Figure \ref{fig:simulation_binary} shows a comparison of populations that have similar competence means but differ in their competence variance.  PCS is only beneficial to populations with high variance (purple and red lines).  Figure \ref{fig: pi_distribution_continuous} present the individual errors distribution of the Flags and Dogs datasets, the Flags dataset has a bimodal distribution, meaning that there are distinct groups of "good" and "bad" workers.  The Dogs dataset has a unimodal distribution.  This might explain why over the Flags dataset PCS reduces the loss better than over the Dogs dataset. \gal{confirm the claim by cut different datasets into high and low variance sub datasets.}

\begin{figure}[ht]
	\centering
 	\includegraphics[trim = 1mm 0mm 16mm 0mm,scale=0.47,clip=true]{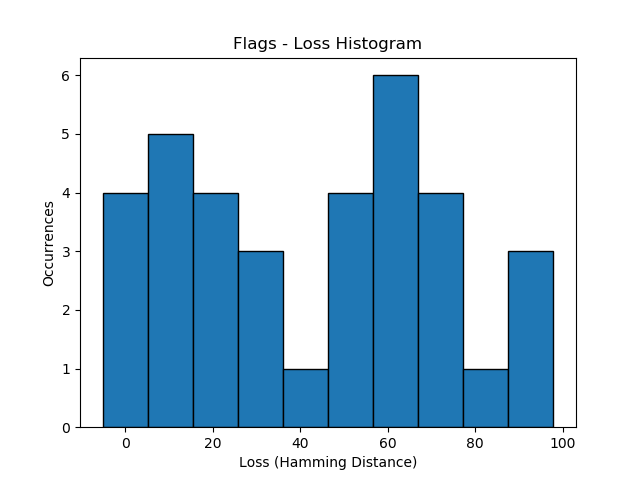}
    \includegraphics[trim = 1mm 0mm 16mm 0mm,scale=0.47,clip=true]{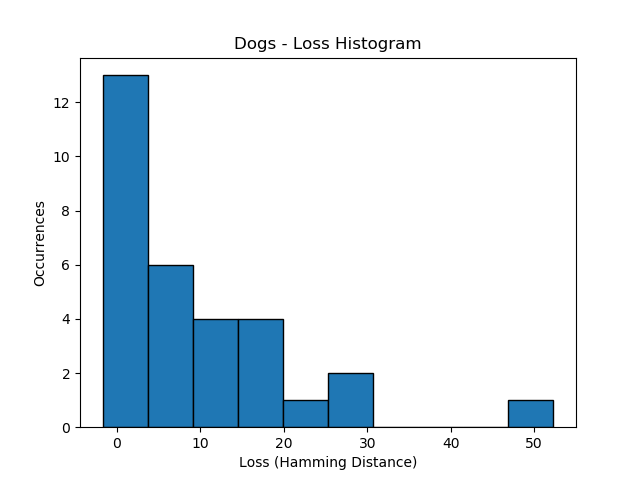}
	\caption{individuals errors histogram.  On the left the Flags dataset shows bimodal distribution of individuals errors, Dogs on the right shows unimodal distribution of individuals errors.
    \label{fig: pi_distribution_continuous}}	
\end{figure}

\end{enumerate}

\section{Explaining Leaders Weights}\label{sec:Explaining leaders Weights}
In the previous section, we observed empirically that competent leaders tend to get more followers and thus higher weight.  We are interested in a theoretical model that explains this phenomenon. One such result was given by Cohensius et al. \cite{cohensius2017proxy} for the limit case of infinite binary questions that are answered by followers $\alpha \cdot k\rightarrow \infty$. In this case, \emph{all} followers select either the best or the worst leader, according to which one is closer. However, in realistic scenarios (including our datasets), the number of questions is typically dozens to hundreds. 

We model a simplified version of the problem, where there is one follower which is requested to choose a leader amongst two possible leaders in a binary domain. Given a worker's competence, we only know the distribution of his answers, and we want to estimate the probability that the follower would choose the better leader.

\paragraph{Bounding the probability of a follower to choose the incompetent leader}
Consider two leaders, High and Low, with judgment competence $p_h,p_l$ respectively, such that $p_h >p_l > 0.5$. In addition, consider a follower Z with judgment competence $p_z > 0.5$.  W.l.o.g. $\vec x^*=\vec 1$. Thus $\vec x_h,\vec x_l$ and $\tilde {\vec x_z}$ are random binary vectors of length $k$.
For simplicity, denote
\[
\vec s_h=\vec x_h(\top(\tilde {\vec x_z})),\vec s_l=\vec x_l(\top(\tilde {\vec x_z})), \vec s_z = \tilde {\vec x_z}(\top(\tilde {\vec x_z})).
\]
That is, $\vec s_h,\vec s_l$ and $\vec s_z$ are the answers of High, Low and Z to the questions answered by Z. Hence $\vec s_h,\vec s_l$ and $\vec s_z$ are random binary vectors of length $\alpha k$, whose entries are `1' with respective probabilities of $p_h,p_l$ and $p_z$. 
\begin{lemma}\label{lemma:bound Pr(bad)}
\[
\Pr(d(\vec s_h,\vec s_z)\geq d(\vec s_l,\vec s_z)) \leq  \exp{\left(-\frac{\alpha k (p_h-p_l)^2(2p_z-1)^2}{2}\right)}.
\]
\end{lemma}

\begin{proof}
Define the random variable $B_i$ such that 
\[
B_i = 
\begin{cases}
1 & \text{High and Z disagree, Low and Z agree}\\
-1 & \text{High and Z agree, Low and Z disagree}\\
0 & \text{Otherwise}
\end{cases}
\]
where agreement/disagreement is with respect to the $i$-th question. Since the entries of $\vec s_h,\vec s_l$ and $\vec s_z$ are independent, so are $B_i$ and $B_{i'}$ for $i\neq i'$. Notice that
\begin{align*}
\mathbb{E}(B_i)&=\mathbb{E}(B_i|s_z^{(i)}=1)\Pr(s_z^{(i)}=1)+\mathbb{E}(B_i|s_z^{(i)}=0)\Pr(s_z^{(i)}=0) \\
&=p_z\left( (1-p_h)p_l - p_h(1-p_l)   \right)+(1-p_z)\left( p_h(1-p_l)-(1-p_h)p_l   \right) \\
&=(2p_z-1)(p_l-p_h).
\end{align*}
Notice that  $d(\vec s_h,\vec s_z) - d(\vec s_l,\vec s_z)\sim\sum_{i=1}^k B_i$, therefore
\begin{align}
\label{eq:bdgfj}
\Pr(d(\vec s_h,\vec s_z)\geq d(\vec s_l,\vec s_z))&=\Pr(d(\vec s_h,\vec s_z)- d(\vec s_l,\vec s_z)\geq 0)=\Pr\left(\sum_{i=1}^k B_i \geq 0\right)  \nonumber \\
&=\Pr\left(\frac{1}{k}\sum_{i=1}^k B_i -\mathbb{E}(B_i) \geq - \mathbb{E}(B_i) \right) 
\end{align}

By invoking Lemma \ref{lemma:hoeff} on Equation (\ref{eq:bdgfj}), we obtain
the desired result.
\end{proof}
\begin{lemma}[Hoeffding's Inequality]
\label{lemma:hoeff}
Let $B_1\dots B_k$ be i.i.d. r.v. such that $B_i\in[-1,1]$ and $\mathbb E(B_i)=\mu$. For every $t>0$ it holds that
\[
\Pr\left(\frac{1}{k}\sum_{i=1}^k B_i -\mu \geq t\right) \leq \exp{\left(  -\frac{k t^2}{2} \right) }.
\]
\end{lemma}

That is, the probability of Z selecting Low decreases exponentially in the distance between Low's and High's competence levels. Another observation is that if $p_z$ approaches $0.5$ (i.e. an incompetent follower), then the term in Lemma \ref{lemma:bound Pr(bad)} approaches $e^0$, in other words, incompetent followers spread their weight roughly evenly over all leaders, whereas competent followers are substantially more likely to give their weight to a competent leader.

This supports the intuitive argument from \cite{cohensius2017proxy} regarding the ``Anna Karenina principle'' (as good workers are indeed similar to one another), and thus at least partially explains the weight distribution of leaders.

\section{Discussion}
We introduced Proxy Crowdsourcing (PCS), an aggregation method that collects complete vectors and partial vectors of answers.  The method then uses the partial vectors to weigh the complete vectors.  We showed that this method can reduce the expected error in various crowdsourcing domains.  PCS can be used in addition to other aggregation methods.  Our conjecture is that for a wide range of aggregation methods, adding $PCS$ to the aggregation method will improve the outcome.  We observe that indeed this is the case for a wide variety of domains.  Table \ref{Table:results summary} shows that PCS improves the outcome for $Plurality$ and $Mean$ aggregation rules. Similar results were obtained for the $Median$ rule.

\paragraph{Allowing Self Selection}  In the experiment discussed in this paper we preallocated workers into leaders and followers.  However preallocation is not mandatory, a requester can allow workers to choose after a sample of questions whether they wish to leave (followers) or to complete the survey (leaders).  We hypothesize that allowing workers to \textit{self select} their role can further improve the results.  This can be done by using a bonus scheme that motivates the strong workers to become leaders and the weak workers to become followers.  In preliminary experiments we saw that PCS with self selection produce more accurate outcome than self selection alone.

\paragraph{Allowing followers to choose their leader}  Previous work \cite{cohensius2017proxy}  assumed that when given the option to choose their leader, followers follow the leader closest  to them.  In this work we wanted to test this assumption, by allowing followers to choose which leader they prefer out of a predefined set of leaders.  We verified that indeed in crowdsourcing tasks, when workers are paid by the number of correct answers, followers tend to follow the leader closest to them.  In an early setting of our experiment, followers were asked: "Why did you choose to follow this leader?"  We observe two types of answers:
\begin{itemize}
\item The random follower: "I was unsure about my choices"
\item The closest leader follower: "Choosing someone who guessed like I did would save time, and would likely give similar results."
\end{itemize}
This verifies our assumption that followers tend to follow the closest proxy.

	\bibliography{proxy}
	\bibliographystyle{plain}	
	\begin{contact}
		Gal Cohensius\\
		Technion---Israel Institute of Technology\\
		Technion, Israel\\
		\email{Gal.Cohensius@gmail.com}
	\end{contact}
	
	\begin{contact}
		Omer Ben - Porat\\
		Technion---Israel Institute of Technology\\
		Technion, Israel\\
		\email{omerbp@gmail.com}
	\end{contact}

	\begin{contact}
		Reshef Meir\\
		Technion---Israel Institute of Technology\\
		Technion, Israel\\
		\email{Reshef.Meir@ie.technion.ac.il}
	\end{contact}
	
	\begin{contact}
		Ofra Amir\\
		Technion---Israel Institute of Technology\\
		Technion, Israel\\
		\email{ofra.amir@gmail.com}
	\end{contact}
\end{document}

%% file: defs.tex
\newtheorem{theorem}{Theorem}

\newtheorem{lemma}[theorem]{Lemma}

\newtheorem{example}[theorem]{Example}

\newcommand{\sproof}{\noindent{\bf Proof.}\hspace*{1em}}

\def\literalqed{{\ \nolinebreak\hfill\mbox{\quad}}}
\long\def\symbolfootnote[#1]#2{\begingroup%
\def\thefootnote{\fnsymbol{footnote}}\footnote[#1]{#2}\endgroup} 

    \makeatletter
    \renewcommand\part{%
      \if@openright
        \cleardoublepage
      \else
        \clearpage
      \fi
      \thispagestyle{empty}%
      \if@twocolumn
        \onecolumn
        \@tempswatrue
      \else
        \@tempswafalse
      \fi
      \null\vfil
      \secdef\@part\@spart}
    \makeatother

\def\calW{{\cal W}}
\def\calX{{\cal X}}

\DeclareMathOperator{\argmin}{\text{argmin}}

\newcommand{\ind}[1]{\llbracket #1 \rrbracket}

\newcommand{\floor}[1]{\left\lfloor #1 \right\rfloor}

\def\({\left(}
\def\){\right)}

\renewcommand{\ind}[1]{\left\llbracket #1 \right\rrbracket}
\newcommand{\xqed}{\mbox{\raggedright $\Diamond$}}
\def\xqedhere{\hfill\xqed}

\renewcommand{\vec}{\boldsymbol}
\def\cal{\mathcal}

\def\wmn{\mathbf{wmn}}

\def\wpl{\mathbf{wpl}}